\documentclass[paper]{ieice}
\usepackage{graphicx}
\usepackage{latexsym}
\usepackage{amsthm}
\usepackage[fleqn]{amsmath}
\usepackage[varg]{txfonts}
\newcommand{\memo}[1]  %memo
{\bigskip \par\noindent
{\bf ¦F}
\fbox{\parbox[t]{7cm} 
{#1}
} \par\bigskip\noindent
}
\newcommand{\minimemo}[1]  % minimemo
{\bigskip \par\noindent
{\bf ¦F}
\fbox{\parbox[t]{4cm} 
{#1}
} \par\bigskip\noindent
}
\newcommand{\memovar}[2]  % memovar
{\bigskip \par\noindent
\fbox{\parbox[t]{#1 cm} 

{#2}
} \par\bigskip\noindent
}

\field{some}
\vol{some}
\no{some}
\SpecialSection{Information Theory and Its Applications}  %\label{sec:1}
\title{A Fundamental Inequality for Lower-bounding the Error Probability for Classical and Quantum Multiple Access Channels and Its Applications}
\authorlist{% fill arguments of \authorentry, otherwise error will be caused. 
 \authorentry{Takuya KUBO}{n}{labA}
 \authorentry{Hiroshi Nagaoka}{m}{labA}
}
\affiliate[labA]{The authors are with the Graduate School of Information Systems, The University of Electro-Communications.}
\received{2015}{2}{13}
\revised{2015}{2}{13}
\newtheorem{theo}{Theorem}

\newtheorem{cor}{Corollary}

\theoremstyle{definition}
\newtheorem{dfn}{Definition}

\theoremstyle{remark}

\def \Pe {{\rm Pe}}
\def \Pr {{\rm Pr}}
\def \Tr {{\rm Tr}}
\def \R  {\mathbb{R}}
\def \H  {\mathcal{H}}
\def \M  {\mathcal{M}}

\def \S  {\mathcal{S}}
\def \X  {\mathcal{X}}
\def \Y  {\mathcal{Y}}
\begin{document}
\maketitle
\begin{summary}
				In the study of the capacity problem for multiple access channels (MACs), a lower bound on the 
				error probability obtained by Han plays a crucial role in the converse parts of several kinds of channel coding theorems in the information-spectrum framework. 
				Recently, Yagi and Oohama showed a tighter bound than the Han bound by means of Polyanskiy's converse. 
				In this paper, we give a new bound which generalizes and strengthens the Yagi-Oohama bound, and demonstrate that the bound plays a fundamental role in deriving extensions of several known bounds. In particular, the Yagi-Oohama bound is generalized to two different directions; i.e, to general input distributions and to general encoders. 
				 In addition we extend these bounds to the quantum MACs and apply them to the converse problems for several information-spectrum settings. 
\end{summary}

\section{Introduction} \label{sec:1}
The capacity problem for multiple access channels(MACs) has been an important topic since Shannon \cite{Shannon} studied it.
This problem is studied for several kinds of settings. 
For instance, in the classical case, Ahlswede \cite{Ahl} found the single-letterized capacity region for stationary and memoryless channels, Han \cite{Han}\cite{Hanbook} 
found the capacity region for the general channels by means of information spectrum method, 
and Winter \cite{Winter} found that for stationary and memoryless channels in the quantum case.
However, there remain some fundamental problems to be solved, including the exponential convergence 
of the error probability in the strong-converse region for stationary memoryless channels and the general information-spectrum formula for the capacity region 
in the quantum case.  So we still need to look for good lower bounds on the error probability. 

In this paper, we discuss lower bounds on the error probability for the following three settings, which are similar but slightly different from each other. 
\begin{itemize}
				\item Setting 1\\
								Let $\X_1$, $\X_2$ and $\Y$ be arbitrary discrete sets 
								on which 
								an input distribution $p(x_1, x_2) $
								and a channel $W(y | x_1, x_2)$ 
								are given. 
								For  a reversed channel $g(x_1,x_2|y)$, which means the probability of decoding (or estimating) the input $(x_1,x_2)$ from the observed output $y$,  
								the error probability is defined by
				\begin{align}
								\hspace{-2em}\Pe(g) := 1 - \sum_{x_1, x_2, y}p(x_1,x_2) W(y | x_1, x_2) g(x_1,x_2|y).  \label{eq:1}
				\end{align}
				\item Setting 2\\ Let $\X_1$, $\X_2$ and $\Y$ be arbitrary discrete sets on which a channel $W(y | x_1, x_2)$ is given.
							Given a pair of message sets $\M_1$ and $\M_2$ with $|\M_1| = M_1$ and $|\M_2| = M_2$ together with encoders $f_1(x_1|m_1)$ and $f_2(x_2|m_2)$, which means the probabilities of encoding the message $m_1$ and $ m_2$ to the inputs $x_1$ and $x_2$ respectively, 
we define the error probability for an arbitrary decoder $g(m_1,m_2|y) $ by 
				\begin{align}
								&\hspace{-2em}\Pe(g)  := 1 - \sum_{m_1,m_2}\frac{1}{M_1 M_2}\notag \\
								&\hspace{-2em}\cdot\sum_{x_1,x_2,y} f_1(x_1|m_1)f_2(x_2|m_2)W(y |x_1, x_2)g(m_1,m_2|y). \label{eq:2}
				\end{align}
				\item Setting 3\\
								Let $\X_1$, $\X_2$ and $\Y$ be arbitrary discrete set s on which a channel $W(y | x_1, x_2)$ is given. 
								Given a pair of codebooks ${\cal C}_1\subset\X_1$ and ${\cal C}_2\subset\X_2$ with  $|{\cal C}_1| = M_1$ and  $|{\cal C}_2| = M_2$, 
								we define the error probability for an arbitrary decoder $g(m_1,m_2|y) $ by 
								
				\begin{align}
								&\Pe(g) :=  \notag \\
								&1 - \frac{1}{M_1 M_2}\sum_{x_1\in{\cal C}_1, x_2\in{\cal C}_2, y}W(y|x_1,x_2)g(x_1,x_2|y). \label{eq:3}
				\end{align}
\end{itemize}

Note that Setting 3 can be regarded as special cases of both Setting~1 and Setting~2.  That is, Setting~3 is 
obtained by restricting $p(x_1, x_2)$ to the product of the uniform distributions on the codebooks in Setting~1, 
and is obtained by restricting encoders $f_1$, $f_2$ to be deterministic and injective in Setting~2. 
In the study of the capacity problem,  Setting~3 have been mainly dealt with so far, as mentioned below for \cite{Han}\cite{Hanbook} and \cite{YO}, while Poor and Verd\'u \cite{PV} 
discussed a lower bound of the error probability in Setting~1 and Polyanskiy \cite{Pol} used Setting~2 in his meta-converse argument. 

In Setting 3, Han \cite{Han}\cite{Hanbook} showed the following lower bound, which is known as the Han bound.
For an arbitrary positive number $\gamma$, it holds that
\begin{align}
				&\Pe(g) \geq \Pr\{(X_1, X_2, Y)\in L_1 \cup L_2 \cup L_3\} - 3\gamma,\label{eq:4}
\end{align}
where  $\Pr$ denotes the probability defined by the joint distribution $p(x_1, x_2, y) = 
p_{u, 1}(x_1) p_{u, 2}(x_2) W(y|x_1, x_2)$ for 
the uniform distributions $p_{u, 1}$ and $p_{u, 2}$ on the codebooks, and 
\begin{align}
				&L_1 := \{ (x_1, x_2, y)| W(y |x_1,x_2) \leq \gamma M_1 p(y | x_2)\}, \label{eq:5}\\
				&L_2 := \{ (x_1, x_2, y)| W(y |x_1,x_2) \leq \gamma M_2 p(y | x_1)\}, \label{eq:6}\\
				&L_3 := \{ (x_1, x_2, y)| W(y |x_1,x_2) \leq \gamma M_3 p(y)\},  \label{eq:7}\\
				&M_3 := M_1M_2. \label{eq:8}
\end{align}
This bound is a MAC extension of the Verd\'u-Han bound \cite{VH} and plays a crucial role in the converse parts of several coding theorems for general MAC channels.

Recently Yagi and Oohama \cite{YO} showed a tighter bound as follows. 
For an arbitrary conditional distribution $q(y |x_1,x_2)$, an arbitrary distribution $\pi$ on $\{1,2,3\}$, and an arbitrary positive number $\gamma'$, it holds that
\begin{align}
				\Pe(g) \geq \Pr\{(X_1, X_2, Y)\in\tilde{L}\} - \gamma' \sum_{i = 1}^3 \frac{\pi_i}{M_i}, \label{eq:9}
\end{align}
where
\begin{align}
				&\tilde{L} :=  \{ (x_1, x_2, y)| W(y |x_1,x_2) \leq \gamma' \tilde{q}(y | x_1, x_2)\}, \label{eq:10}\\
				&\tilde{q}(y|x_1,x_2) = \pi_1 q(y |x_2) + \pi_2 q(y|x_1) + \pi_3 q(y),\label{eq:11}
\end{align} 
and $q(y |x_1), q(y |x_2)$ and $q(y)$ are the conditional and marginal distributions defined from the joint distribution 
\begin{align}
				q(x_1,x_2,y) = p_{u,1}(x_1)p_{u,2}(x_2)q(y|x_1,x_2). \label{eq:12}
\end{align} 
If we set $\pi_i = \frac{M_i}{\sum_j M_j}$, $\gamma' = \gamma \sum_j M_j$ and $q(y|x_1, x_2)  = p(y|x_1, x_2) $, 
we can rewrite \eqref{eq:9} and \eqref{eq:10} as follows.
\begin{align}
				&\hspace{-2em}\Pe(g) \geq \Pr\{(X_1, X_2, Y)\in\tilde{L}\} - 3 \gamma,\label{eq:13}\\
				&\hspace{-2em}\tilde{L} = \{ (x_1, x_2, y)|\notag \\
				&\hspace{-2em}W(y |x_1,x_2) \leq \gamma (M_1 p(y |x_2) + M_2 p(y|x_1) + M_3 p(y))\}. \label{eq:14}
\end{align}
Since $L_1 \cup L_2 \cup L_3 \subset \tilde{L}$, \eqref{eq:9} is tighter than \eqref{eq:4}.

In what follows, we first show an extension of the Yagi-Oohama bound \eqref{eq:9} as Theorem~\ref{theo:1}  in section~\ref{sec:2}, where 
the Yagi-Oohama bound is extended from Setting~3 to Setting~1 and, in addition, is slightly strengthened as is seen in subsection~\ref{sec:4}. 
We also see in subsection~\ref{sec:5} 
that 
the theorem  yields a MAC version of the Poor-Verd\'{u} bound.  
In section~\ref{sec:6}, we use Theorem~\ref{theo:1} again to obtain an extension of 
the Yagi-Oohama bound to Setting~2. 
In section~\ref{sec:7}, we show that 
these results are naturally extended to the quantum case. Lastly in section~\ref{sec:8}, we apply 
them to obtain  some asymptotic results which correspond to the 
converse parts of the general capacity theorems obtained by Han in the classical case.
Concluding remarks are given in section~\ref{sec:9}. 
\section{A fundamental inequality on the error probability for the classical MACs} \label{sec:2}
The following inequality plays a fundamental role in this paper. 
\begin{theo} \label{theo:1} In Setting~1 given in section \ref{sec:1}, 
								for an arbitrary decoder $g$, arbitrary $\alpha_{1},\alpha_{2},\alpha_{3} \geq 0$, an arbitrary probability distribution $q(y)$ on $\Y$, and arbitrary nonnegative-valued functions $q_1(x_1,y), q_2(x_2,y)$ satisfying that $q(y) \geq q_1(x_1,y)$ and $q(y) \geq q_2(x_2,y)\;(\forall x_1, x_2, y)$, we have
				\begin{align}
								1 - \Pe(g) - \sum_i \alpha_i \leq \sum_{x_1,x_2,y}[p(x_1,x_2.y) - q_\alpha(x_1,x_2,y)]_+, \label{eq:15}
				\end{align}
				where
				\begin{align}
								&q_\alpha(x_1,x_2,y) = \alpha_1 q_2(x_2, y) + \alpha_2 q_1(x_1, y) + \alpha_3 q(y),\label{eq:16}\\
								&[t]_+ = \max\{0,t\}. \ \ (t \in \R) \label{eq:17}
				\end{align}
				
\end{theo}
\begin{proof}
				As in the proof of Neyman-Pearson's Lemma, it follows from   $0\leq g(x_1, x_2 | y) \leq 1$ that 
				\begin{align}
								\sum_{x_1,x_2,y}&[(p(x_1,x_2.y) - q_\alpha(x_1,x_2,y)]_+ \notag \\
								&\geq \sum_{x_1,x_2,y}\{(p(x_1,x_2.y) - q_\alpha(x_1,x_2,y)\}g(x_1,x_2|y) \notag\\
								&\geq 1 - \Pe(g) - \sum_i \alpha_i \sum_{x_1,x_2,y} q(y)g(x_1,x_2|y) \notag \\
								& = 1 -  \Pe(g) -  \sum_i \alpha_i ,  \label{eq:18}
				\end{align}
				where the second inequality follows from  $q(y) \geq q_1(x_1,y)$ and $q(y) \geq q_2(x_2,y)$.
\end{proof}
\section{Corollaries of Theorem~\ref{theo:1} in Setting~1}  \label{sec:3}
\subsection{A Yagi-Oohama-type bound} \label{sec:4}
The Yagi-Oohama bound is extended to the general input distributions in the following form. 
\begin{cor} \label{cor:1}
				\begin{align}
								\Pe(g) \geq \Pr\{p(X_1,X_2,Y) \leq q_\alpha(X_1,X_2,Y)\} -\sum_i\alpha_i  \label{eq:19}
				\end{align}
\end{cor}
\begin{proof} Eq.\ \eqref{eq:19} immediately follows from \eqref{eq:15}, since 
				\begin{align}
								&\sum_{x_1,x_2,y}[p(x_1,x_2.y) - q_\alpha(x_1,x_2,y)]_+ \notag \\
								&= \sum_{x_1,x_2,y}(p(x_1,x_2.y) - q_\alpha(x_1,x_2,y))\notag\\
								&\hspace{4em}\cdot1\{p(x_1,x_2.y) > q_\alpha(x_1,x_2,y)\}\notag\\
								&\leq \sum_{x_1,x_2,y}p(x_1,x_2.y)\;1\{p(x_1,x_2.y) > q_\alpha(x_1,x_2,y)\}\notag \\
								&=1- \Pr\{p(X_1,X_2,Y) \leq q_\alpha(X_1,X_2,Y)\} ,  \label{eq:20}
				\end{align}
				where $1\{\ \}$ is the indicator function.
\end{proof}
In Setting~3, the original Yagi-Oohama bound  \eqref{eq:13} is obtained from \eqref{eq:19} by setting 
 $\alpha_i = \gamma'\pi_i/M_i$. 
\subsection{A Poor-Verd\'{u}-type bound} \label{sec:5}
While Corollary \ref{cor:1} can be regarded as a MAC extension of the Verd\'u-Han bound \cite{VH} (or the Hayashi-Nagaoka bound \cite{HN} in the sense that arbitrary output distributions are allowed), 
the following bound corresponds to  the Poor-Verd\'u bound \cite{PV}.
\begin{cor} \label{cor:2}
				\begin{align}
								\hspace{-2em}\Pe(g) \geq \Bigl(1 -\sum_i\alpha_i\Bigr)\, \Pr\{p(X_1,X_2,Y) \leq p_\alpha(X_1,X_2,Y)\}, \label{eq:21}
				\end{align}
				where 
				\begin{align}
								p_\alpha(x_1,x_2,y) = \alpha_1 p(x_2, y) + \alpha_2 p(x_1, y) + \alpha_3 p(y),\label{eq:22}
				\end{align}
				and $p(x_1, y), p(x_2,y)$ and $p(y)$ are marginal distributions defined from the joint distribution $p(x_1,x_2,y)$.
\end{cor}
\begin{proof}
				If $q = p$, the right hand side of \eqref{eq:15} is rewritten as 
				\begin{align}
								&\sum_{x_1,x_2,y}(p(x_1,x_2.y) - p_\alpha(x_1,x_2,y))\notag \\
								&\hspace{6em}\cdot1\{p(x_1,x_2.y) > p_\alpha(x_1,x_2,y)\}\notag\\
								&\leq (1 - \sum_i\alpha_i)\sum_{x_1,x_2,y}p(x_1,x_2.y)\notag\\
								&\hspace{6em}1\{p(x_1,x_2.y) > p_\alpha(x_1,x_2,y)\},\label{eq:23}
				\end{align}
				where the inequality follows from $p(x_1,x_2,y) \leq p(y)$, $p(x_1,x_2,y) \leq p(x_1, y)$, and $p(x_1,x_2,y) \leq p(x_1, y)$.
\end{proof}
\section{Corollaries of Theorem~\ref{theo:1} in Setting~2}  \label{sec:6}
An extension of the Yagi-Oohama bound to Setting~2, where general stochastic encoders are allowed, 
is also derived from Theorem~\ref{theo:1} as follows. 
\begin{cor}  \label{cor:3}
				In Setting~2,   
				for an arbitrary decoder $g$, arbitrary $\gamma'_{1},\gamma'_{2},\gamma'_{3} \geq 0$, an arbitrary distribution $q$ on $\Y$, and arbitrary conditional distributions $q_1(y|x_1), q_2(y|x_2)$ satisfying that
				\begin{align}
								q(y) &\geq q'_1(m_1,y) := \frac{1}{M_1} \sum_{x_1} f_1(x_1|m_1)q_1(y |x_1), \label{eq:24}\\	
								q(y) &\geq q'_2(m_2,y) := \frac{1}{M_2} \sum_{x_2} f_2(x_2|m_2)q_2(y |x_2), \label{eq:25}\\
								&\hspace{14em}(\forall m_1,m_2,y)\notag
				\end{align}
				we have
				\begin{align}
								&\hspace{-1em}1 - \Pe(g) -\sum_i \frac{\gamma'_i}{M_i}\notag\\
								&\hspace{-1em}\geq \sum_{x_1,x_2,y}p_1(x_1)p_2(x_2)\left[W(y|x_1,x_2) - \tilde{q}_{\gamma'}(y|x_1,x_2)\right]_+, \label{eq:26}
				\end{align}
				where
\begin{align}
				\tilde{q}_{\gamma'}(y|x_1,x_2) &= \gamma'_1 q(y |x_2) + \gamma'_2 q(y|x_1) + \gamma'_3 q(y), \label{eq:27}\\
				p_1(x_1) &= \sum_{m_1} \frac{1}{M_1}f_1(x_1|m_1), \label{eq:28}\\
				p_2(x_2) &= \sum_{m_2} \frac{1}{M_2}f_1(x_2|m_2), \label{eq:29}\\
								M_3 &= M_1M_2. \label{eq:30}
\end{align} 

\end{cor}
\begin{proof}
Let a channel $V$ from $\M_1\times\M_2$ to $\Y$ be defined by
				\begin{align}
								V(y|m_1,m_2) = \sum_{x_1,x_2}f_1(x_1|m_1)f_2(x_2|m_2)W(y|x_1,x_2). \label{eq:31}
				\end{align}
				Then, replacing $\X_1$, $\X_2$ and $W$ with $\M_1$, $\M_2$, and $V$ in Theorem \ref{theo:1} and letting the input distribution be uniform on $\M_1 \times \M_2$, we have
				\begin{align}
								&1 - \Pe(g) -\sum_i\frac{\gamma'_i}{M_i}\notag \\
								&\leq \sum_{m_1,m_2,y}\left[\frac{1}{M_1M_2}V(y|m_1,m_2) - \tilde{q}'_{\gamma'}(m_1,m_2,y)\right]_+ ,\label{eq:32}
				\end{align}
				where
				\begin{align}
							\hspace{-1em}	\tilde{q}'_{\gamma'}(m_1,m_2,y) &= \frac{\gamma_1'}{M_1}q'_2(m_2,y)+\frac{\gamma'_2}{M_2}q'_1(m_1,y)+\frac{\gamma'_3}{M_3}q(y). \label{eq:33}
				\end{align}
				From the convexity of $t\rightarrow[t]_+$, we have
				\begin{align}
								1 - &\Pe(g) -\sum_i\frac{\gamma'_i}{M_i}\notag \\
								&\leq \sum_{m_1,m_2,x_1,x_2,y}\frac{1}{M_1M_2}f_1(x_1|m_1)f_2(x_2|m_2)\notag \\
								&\hspace{6em}\cdot [W(y|x_1,x_2) - \tilde{q}_{\gamma'}(y|x_1,x_2)]_+ .\label{eq:34}
				\end{align}
\end{proof}
This inequality immediately derives the following bound, which is the direct extension of the Yagi-Oohama bound to Setting~2.
\begin{cor} \label{cor:4}
				In Setting 2, for an arbitrary decoder $g$, an arbitrary distribution $\pi$ on $\{1,2,3\}$, an arbitrary number $\gamma \geq 0$, and an arbitrary channel $q(y|x_1,x_2)$, we have
				\begin{align}
								\hspace{-1.5em}\Pe(g)\geq \Pr\{W(Y|X_1,X_2) \leq \gamma \tilde{q}(Y|X_1,X_2)\} - \gamma\sum_i\frac{\pi_i}{M_i},	 \label{eq:35}
				\end{align}
				where the random variables $X_1$, $X_2$, and $Y$ are defined by the joint distribution 
				\begin{align}
								\hspace{-1em}p(x_1,x_2,y) = \frac{1}{M_1M_2}\sum_{m_1,m_2}f_1(x_1|m_1)f_2(x_2|m_2)W(y|x_1,x_2), \label{eq:36}
				\end{align}
				and $\tilde{q}$ is defined by \eqref{eq:11}.
\end{cor}
\section{Lower bounds on the error probability for the quantum MACs} \label{sec:7}
In this section we extend the arguments of  previous sections 
to classical-quantum MACs. 
In the single access case, Hayashi and Nagaoka \cite{HN} extended the Verd\'u-Han bound into the quantum case,
and the present authors \cite{KN1}, \cite{KN2} extended the Poor-Verd\'u bound. 
Applying a similar argument to the ones developed there, we extend Theorem \ref{theo:1} as presented in Theorem \ref{theo:2}, from 
which the corresponding results to Corollaries \ref{cor:1}-\ref{cor:3} immediately follow. 

We begin with rewriting Setting~1 and Setting~2 to the quantum situation. Setting~3 is omitted since it 
is included in Setting~1 and Setting~2. 
%nag
 \begin{itemize}
				\item Setting Q1\\
								Let $\X_1$, $\X_2$ be arbitrary discrete sets on which an input distribution $p(x_1, x_2) $ is given.
								Let $\H$ be an arbitrary Hilbert space and $\S(\H)$ be the set of density operators on $\H$ and $W: \X_1 \times \X_2 \to \S(\H)$ be a classical-quantum channel (a quantum channel, for short).
								When a POVM (Positive Operator-Valued Measure) $Y = \{Y_{x_1,x_2}\}$, which satisfies that $\sum_{x_1,x_2}Y_{x_1,x_2} = I$ and $Y_{x_1,x_2} \geq 0\;\;(\forall x_1,x_2)$, represents a decoding (or estimating) process, the error probability is defined by
								\begin{align}
												\Pe(Y) := 1 - \sum_{x_1,x_2} p(x_1,x_2) \Tr[W_{x_1,x_2}Y_{x_1,x_2}]. \label{eq:37}
								\end{align}
				\item Setting Q2\\Let $\X_1$, $\X_2$ be arbitrary discrete sets, $\H$ be an arbitrary Hilbert space and a quantum channel $W: \X_1 \times \X_2 \to \S(\H)$ is given.
								As in Setting 2, given a pair of message sets $\M_1$ and $\M_2$ with $|\M_1| = M_1$ and $|\M_2| = M_2$ together with encoders $f_1(x_1|m_1)$ and $f_2(x_2|m_2)$, which means the probabilities of encoding the message $m_1$ and $ m_2$ to the inputs $x_1$ and $x_2$ respectively, 
we define the error probability for an arbitrary POVM $Y$ whose indexes are in $\M_1 \times \M_2$ by
				\begin{align}
								\Pe&(Y)  := 1 - \sum_{m_1,m_2}\frac{1}{M_1 M_2}\notag \\
								&\cdot \sum_{x_1,x_2} f_1(x_1|m_1)f_2(x_2|m_2)\Tr[W_{x_1, x_2}Y_{m_1,m_2}]. \label{eq:38}
				\end{align}
\end{itemize}
%nag
%\memo{'±'±'ÅŽžŠÔØ'êBˆÈ~'͌ÓTŒn'Ì Theorem 'Æ Corollary 'Ì'·'ׂĂ̗ʎq"Å'ð'WX'Əq'ׂĂ¢'¯'΂悢B—]Œv'Èà–¾'Í'Å'«'邾'¯È'«AŒÃ"TŒn'Ƃ̑ΉžŠÖŒW'ª'­'Á'«'è'·'é—l'ɏ''­'Ì'ª'æ'¢'ÆŽv'¤B}

%Let $\X_1$ and $\X_2$ be finite sets and suppose that a distribution $p(x_1,x_2)$ on $\X_1 \times \X_2$ is given.
%Let $\H$ be a Hilbert space and $\S(\H)$ be the state set on $\H$, which is the set of operators on $\H$ which are positive semidefinite and whose trace equal to unity.
%Suppose that a map $W:\X_1 \times \X_2 \to \S(\H)$, which is called  a classical-quantum channel(quantum channel, for short), is given.

%Here, a decoder is represented by POVM(Positive Operator Valued Measure) $Y = \{Y_{x_1,x_2}\}$ satisfying that $Y_{x_1,x_2} \geq 0 (\forall x_1, x_2)$ and $\sum_{x_1,x_2}Y_{x_1,x_2} \leq I$.
%The error probability is defined as follows.
%\begin{dfn} \label{dfn:1}
				%\begin{align}
							%\Pe(Y)  := 1 - \sum_{x_1,x_2} p(x_1, x_2) \Tr [W_{x_1,x_2}Y_{x_1,x_2}]. \label{eq:41}
			%\end{align}
%\end{dfn}
%On the error probability, the following inequality holds.
Theorem \ref{theo:1} is extended as follows.
\begin{theo} \label{theo:2}
				In Setting~Q1,
				for an arbitrary POVM $Y$, arbitrary $\alpha_{1},\alpha_{2},\alpha_{3} \geq 0$, an arbitrary density operator $\sigma$ on $\H$, and arbitrary positive semidefinite operators $\sigma_{x_1}, \sigma_{x_2}$ satisfying that $\sigma \geq \sigma_{x_1}$ and $\sigma \geq \sigma_{x_2}\;(\forall x_1, x_2)$, we have
				\begin{align}
								\hspace{-1em}1 - \Pe(Y) - \sum_i \alpha_i \leq \sum_{x_1,x_2}\Tr[(p(x_1,x_2)W_{x_1,x_2} - \sigma_{\alpha,x_1,x_2})_+], \label{eq:39}
				\end{align}
				where
				\begin{align}
								&\sigma_{\alpha, x_1,x_2} = \alpha_1 \sigma_{x_2} + \alpha_2 \sigma_{x_1} + \alpha_3 \sigma,\label{eq:40}\\
								&A_+ := A\{A \geq 0\}.\label{eq:41}
				\end{align}
Here and in the sequel, we use the notation $\{A\leq B\} = \{B\geq A\}$ to mean a projector on $\H$ which is defined as follows.
When A - B is spectrum-decomposed as  
\begin{align}
				A - B &= \sum_i \lambda_i E_i, \label{eq:42}\\
				\{A\leq B\} &:= \sum_{i: \lambda_i \leq 0} E_i. \label{eq:43}
\end{align}
\end{theo}
%Its proof is almost the same as the classical one.
\begin{proof}
				As in the classical case, it follows from $0 \leq Y_{x_1,x_2}\leq I$ that
				\begin{align}
								&\sum_{x_1,x_2}\Tr[(p(x_1,x_2)W_{x_1,x_2} - \sigma_{\alpha,x_1,x_2})_+] \notag \\
								&\geq \sum_{x_1,x_2}\Tr[(p(x_1,x_2)W_{x_1,x_2} - \sigma_{\alpha,x_1,x_2})Y_{x_1,x_2}] \notag\\
								&\geq 1 - \Pe(Y) - \sum_i \alpha_i \sum_{x_1,x_2} \Tr[\sigma Y_{x_1,x_2}]  \notag\\
								&= 1 - \Pe(Y) - \sum_i \alpha_i, \label{eq:44}
				\end{align}
				where the second inequality follows from  $\sigma \geq \sigma_{x_1}$ and $\sigma \geq \sigma_{x_2}$.
\end{proof}
Obviously, as Theorem \ref{theo:1} derives Corollaries \ref{cor:1} and \ref{cor:2}, Theorem \ref{theo:2} derives the following corollaries.
%Obviously, Theorem \ref{theo:2} implies the quantum extensions of Corollaries \ref{cor:1} and \ref{cor:2} as follows.
\begin{cor} \label{cor:5}
				\begin{align}
								&\Pe(Y) \notag\\
								&\geq \sum_{x_1,x_2} p(x_1,x_2)\Tr[W_{x_1,x_2}\{p(x_1,x_2)W_{x_1,x_2}\leq \sigma_{\alpha, x_1,x_2}\}] \notag \\
								&\hspace{14em}- \sum_i \alpha_i\label{eq:45}
				\end{align}
\end{cor}
\begin{cor} \label{cor:6}
				\begin{align}
								&\Pe(Y) \geq (1 -\sum_i \alpha_i)\notag \\
								&\cdot\sum_{x_1,x_2} p(x_1,x_2)\Tr[W_{x_1,x_2}\{p(x_1,x_2)W_{x_1,x_2}\leq W_{\alpha, x_1,x_2}\}],  \label{eq:46}
				\end{align}
								where
								\begin{align}
												W_{\alpha, x_1,x_2} &= \alpha_1W_{p,x_2} + \alpha_2 W_{x_1,p} + \alpha_3 W_p, \label{eq:47}\\
								W_p &:= \sum_{x_1,x_2}p(x_1,x_2)W_{x_1,x_2},\label{eq:48}\\
								 W_{p,x_2} &:= \sum_{x_1}p(x_1,x_2)W_{x_1,x_2},\label{eq:49}\\
								 W_{x_1,p} &:= \sum_{x_2}p(x_1,x_2)W_{x_1,x_2}. \label{eq:50}
				\end{align}

\end{cor}

Corollary \ref{cor:5} is a MAC extension of the Hayashi-Nagaoka bound, and Corollary \ref{cor:6} is a quantum MAC extension of the Poor-Verd\'u bound.

%Next, we introduce the quantum version of  Setting~2 and show Theorem \ref{theo:3} can be extended into the quantum case.
%The quantum version of Setting~2 is obtained by modifying the classical setting in the same way as we did for Setting~1. 
%That is, a quantum channel $W$ is given instead of a classical channel, and a decoder is represented by POVM $Y$ which is indexed by $\M_1 \times \M_2$ instead of conditional distribution $g$.
%In this setting, the error probability is defined as
%\begin{dfn} \label{dfn:1}
%\begin{align}
%\Pe&(Y)  := 1 - \sum_{m_1,m_2}\frac{1}{M_1 M_2}\notag \\
		%
	
%&\cdot \sum_{x_1,x_2} f_1(x_1|m_1)f_2(x_2|m_2)\Tr[ W_{x_1,x_2}Y_{m_1,m_2}]. \label{eq:54}
%				\end{align}
%\end{dfn}
Corollary \ref{cor:3} is also extended to the following, which can be proved almost in parallel with the classical one, noting that the convexity of $t\mapsto [t]_+$ should be replaced with the convexity of $A \mapsto {\rm Tr}[A_+]$.
\begin{cor} \label{cor:7}
				In Setting~Q2,
				for an arbitrary POVM $Y$, arbitrary $\gamma'_{1},\gamma'_{2},\gamma'_{3} \geq 0$, an arbitrary density operator $\sigma$ on $\H$, and arbitrary density operators $\sigma_{x_1}, \sigma_{x_2}$ satisfying that
				\begin{align}
								\sigma &\geq \sigma'_{m_1} := \frac{1}{M_1} \sum_{x_1} f_1(x_1|m_1)\sigma_{x_1},\label{eq:51} \\	
								\sigma &\geq \sigma'_{m_2} := \frac{1}{M_2} \sum_{x_2} f_2(x_2|m_2)\sigma_{x_2}, \label{eq:52}\\	
								&\hspace{14em}(\forall m_1,m_2)\notag
				\end{align}
				we have
				\begin{align}
								&1 -\Pe(Y) - \sum_i \frac{\gamma'_i}{M_i}\notag\\
								&\leq \sum_{x_1,x_2}p_1(x_1)p_2(x_2)\Tr[(W_{x_1,x_2} - \tilde{\sigma}_{\gamma', x_1, x_2})_+],\label{eq:53}
				\end{align}
				where
\begin{align}
				\tilde{\sigma}_{\gamma',x_1,x_2} &= \gamma'_1 \sigma_{x_2} + \gamma'_2 \sigma_{x_1} + \gamma'_3 \sigma, \label{eq:54}\\
				p_1(x_1) &= \sum_{m_1} \frac{1}{M_1}f_1(x_1|m_1), \label{eq:55}\\
				p_2(x_2) &= \sum_{m_2} \frac{1}{M_2}f_2(x_2|m_2). \label{eq:56}
\end{align} 
\end{cor}
%Its proof is essentially the same as the classical one.% but we need to the following lemma 
%corresponding to the convexity of $t\mapsto [t]_+$. 
%Note that the convexity of $t\mapsto [t]_+$ is replaced by 
%\begin{lem} \label{lem:1}
				%For Hermitian operator $A$, a map $A \mapsto {\rm Tr}[A_+]$ is convex.
%\end{lem}
%This lemma is not new but we give a proof for the readers' sake.
%\begin{proof}
				%For arbitrary Hermitian operators $A_1,A_2$ and $\forall p \in [0,1]$, it follows that
				%\begin{align}
								%{\rm Tr}&[(pA_1+(1-p)A_2)_+]\notag\\
								%&= \max_{0\leq T \leq I}{\rm Tr}[(pA_1+(1-p)A_2)T]\notag\\
								%&\leq \max_{0\leq T_1 \leq I}p{\rm Tr}[A_1T_1] + (1-p)\max_{0\leq T_2 \leq I}{\rm Tr}[A_2T_2] \label{eq:57}
				%\end{align}
%\end{proof}
%
\section{Applications of Theorem \ref{theo:2} to the quantum information spectrum setting}  \label{sec:8}
In this section, we show applications of Theorem~\ref{theo:2} to the quantum MAC coding problems;
the converse parts of the $\varepsilon$-capacity region problem and the strong converse region problem, which Han \cite{Han} \cite{Hanbook} showed in the classical case.

Let us introduce the setting of the quantum MAC coding problem.
Let $\vec{\X_1} = \{\X_1^{(n)}\}_{n = 1}^\infty$ and $\vec{\X_2} = \{\X_2^{(n)}\}_{n = 1}^\infty$ be sequences of discrete sets, and $\vec{\H} = \{\H^{(n)}\}_{n = 1}^\infty$ be a sequence of Hilbert spaces, for which a sequence of quantum MACs $\vec{W} = \{W^{(n)}: \X_1^{(n)} \times \X_2^{(n)} \to \S(\H^{(n)})\}_{n = 1}^\infty$ is given.
Suppose that, for each $n$, a pair of encoders and a decoder are given in terms of conditional probability distributions $f_1(x_1^{(n)}|m_1^{(n)})$, $f_2(x_2^{(n)}|m_2^{(n)})$ and a POVM $Y^{(n)} = \{Y^{(n)}_{m_1^{(n)}, m_1^{(n)}}\}$ respectively, where $m_1^{(n)} \in \{1 ,\dots, M_1^{(n)}\},m_2^{(n)} \in\{1,\dots, M_2^{(n)}\}$. 
The error probability is then defined as follows:
				\begin{align}
								\Pe^{(n)}(Y^{(n)})  = 1 - \sum_{m^{(n)}_1,m^{(n)}_2} &\frac{1}{M_1^{(n)}M_2^{(n)}}f_1(x^{(n)}_1|m^{(n)}_1)f_2(x^{(n)}_2|m^{(n)}_2)\notag \\
								&\cdot \Tr [W^{(n)}_{x^{(n)}_1,x^{(n)}_2}Y^{(n)}_{m^{(n)}_1,m^{(n)}_2}].  \label{eq:57}
			\end{align}
Here, we call a triple of encoders and decoder $(f^{(n)}_1,f^{(n)}_2,Y^{(n)})$ whose error probability equals $\varepsilon_n$ an $(n, M_1^{(n)}, M_2^{(n)}, \varepsilon_n)$-code.

Now, we introduce the $\varepsilon$-capacity region $C(\varepsilon|\vec{W})$.
\begin{dfn} \label{dfn:1}
				The $\varepsilon$-capacity region $C(\varepsilon|\vec{W})$ is defined as
				\begin{align}
								C(\varepsilon|\vec{W}) := \{(R_1,R_2)|
												&\exists\{(n, M_1^{(n)}, M_2^{(n)}, \varepsilon_n)\text{-code}\}_{n = 1}^\infty\text{ s.t.}\notag\\ 
								&\limsup_{n \to \infty} \varepsilon_n \leq \varepsilon,\notag\\
								&\liminf_{n \to \infty} \frac{1}{n}\log M_1^{(n)} \geq R_1,\notag\\
								&\liminf_{n \to \infty} \frac{1}{n}\log M_2^{(n)} \geq R_2 \}.\label{eq:58}
				\end{align}
\end{dfn}
We also introduce $C^*(\vec{W})$ which represents the complement of the strong converse region.
\begin{dfn} \label{dfn:2}
				\begin{align}
								C^*(\vec{W}) := \{(R_1,R_2)|
												&\exists\{(n, M_1^{(n)}, M_2^{(n)}, \varepsilon_n)\text{-code}\}_{n = 1}^\infty\text{ s.t.}\notag\\ 
								&\liminf_{n \to \infty} \varepsilon_n < 1,\notag\\
								&\liminf_{n \to \infty} \frac{1}{n}\log M_1^{(n)} \geq R_1,\notag\\
								&\liminf_{n \to \infty} \frac{1}{n}\log M_2^{(n)} \geq R_2 \}.\label{eq:59}
				\end{align}
\end{dfn}
\vspace{1ex}
Next, we introduce the following quantities.
\begin{dfn} \label{dfn:3}
				\begin{align}
								&K(R_1,R_2|\vec{p_1},\vec{p_2},\vec{\sigma}) \notag \\
								&:= \limsup_{n \to \infty} \sum_{x_1^{(n)},x_1^{(n)}}p_1^{(n)}(x_1^{(n)})p_2^{(n)}(x_2^{(n)}) \Tr[W^{(n)}_{x_1^{(n)},x_2^{(n)}}\notag\\
								&\cdot\{W^{(n)}_{x_1^{(n)},x_2^{(n)}}\leq e^{nR_1}\sigma^{(n)}_{x_2^{(n)}}+e^{nR_2}\sigma^{(n)}_{x_1^{(n)}} +e^{n(R_1 + R_2)}\sigma^{(n)}\}],\label{eq:60} \\
								&K^*(R_1,R_2|\vec{p_1},\vec{p_2},\vec{\sigma}) \notag \\
								&:= \liminf_{n \to \infty} \sum_{x_1^{(n)},x_1^{(n)}}p_1^{(n)}(x_1^{(n)})p_2^{(n)}(x_2^{(n)}) \Tr[W^{(n)}_{x_1^{(n)},x_2^{(n)}}\notag\\
								&\cdot\{W^{(n)}_{x_1^{(n)},x_2^{(n)}}\leq e^{nR_1}\sigma^{(n)}_{x_2^{(n)}}+e^{nR_2}\sigma^{(n)}_{x_1^{(n)}} +e^{n(R_1 + R_2)}\sigma^{(n)}\}],\label{eq:61}
				\end{align}
				where $\vec{p_1} = \{p_1^{(n)}\}_{n = 1}^\infty$ and $\vec{p_2} = \{p_2^{(n)}\}_{n = 1}^\infty$ are sequences of probability distributions on $\vec{\X_1}$ and $\vec{\X_2}$, and $\vec{\sigma}$ is a sequence of a triple of density operators$(\sigma^{(n)}, \sigma^{(n)}_{x^{(n)}_1}, \sigma^{(n)}_{x^{(n)}_2})$ which satisfies that 
\begin{align}
		\sigma^{(n)} = \sum_{x^{(n)}_1} p_1^{(n)}( x^{(n)}_1 )\sigma^{(n)}_{x^{(n)}_1}		 \label{eq:62}
\end{align}
and
\begin{align}
		\sigma^{(n)} = \sum_{x^{(n)}_2} p_2^{(n)}( x^{(n)}_2 )\sigma^{(n)}_{x^{(n)}_2}		 \label{eq:63}
\end{align}
for each $n$.

\end{dfn}
With these notations, we have
\begin{theo} \label{theo:3}
				\begin{align}
								\hspace{-2em}C(\varepsilon|\vec{W}) \subset \bigcup_{\vec{p_1},\vec{p_2}} \bigcap_{\vec{\sigma}}\mathrm{Cl}( \{(R_1,R_2)| K(R_1,R_2|\vec{p_1},\vec{p_2},\vec{\sigma}) \leq \varepsilon\}), \label{eq:64}
				\end{align}
				where $\mathrm{Cl}(\cdot)$ denotes the closure operation.
\end{theo}
\begin{proof}
			If $(R_1, R_2) \in C(\varepsilon|	\vec{W})$, then from the difinition of $C(\varepsilon|\vec{W})$ there exists a sequence of $(n, M_1^{(n)}, M_2^{(n)}, \varepsilon_n)$-codes satisfying that
				\begin{align}
								& M_1^{(n)} \geq e^{n (R_1 - \gamma)},\label{eq:65}\\
								& M_2^{(n)} \geq e^{n (R_2 - \gamma)},\label{eq:66}
				\end{align}
				for an arbitrary positive number $\gamma$ and all sufficiently large $n$, and
				\begin{align}
								\limsup_{n \to \infty} \varepsilon_n \leq \varepsilon.\label{eq:67}
				\end{align}
				Using these codes, setting the sequences of the input distributions as
				\begin{align}
								p_1^{(n)}(x_1^{(n)}) &= \frac{1}{M_1^{(n)}} \sum_{m_1^{(n)}}f_1^{(n)}(x_1^{(n)}|m_1^{(n)}),\label{eq:68}\\
								p_2^{(n)}(x_2^{(n)}) &=\frac{1}{M_2^{(n)}} \sum_{m_2^{(n)}}f_2^{(n)}(x_2^{(n)}|m_2^{(n)}).\label{eq:69}
				\end{align}
				Now, from Corollary \ref{cor:7}, for arbitrary $\vec{\sigma}$ satisfying \eqref{eq:62} and \eqref{eq:63}, we have
				\begin{align}
								&1 - \varepsilon_n - 3 e^{-n\gamma} \leq \sum_{x_1^{(n)}, x_2^{(n)}}p_1^{(n)}(x_1^{(n)})p_2^{(n)}(x_2^{(n)})\notag\\
								&\hspace{-2em}\cdot\Tr\left[\left(W^{(n)}_{x_1^{(n)},x_2^{(n)}} - e^{-n\gamma} \left(M_1^{(n)} \sigma^{(n)}_{x_2^{(n)}} + M_2^{(n)} \sigma^{(n)}_{x_1^{(n)}} + M_1^{(n)}M_2^{(n)} \sigma^{(n)}\right)\right)_+\right].\label{eq:70}
				\end{align}
				From \eqref{eq:65}, \eqref{eq:66} and from the fact that for arbitrary Hermitian operators $A,B$, $\Tr[A_+] \leq \Tr[B_+]$ if $A \leq B$, we have
				\begin{align}
								&\hspace{-2em}1 - \varepsilon_n - 3 e^{-n\gamma} \notag\\
								&\hspace{-1em}\leq \sum_{x_1^{(n)}, x_2^{(n)}}p_1^{(n)}(x_1^{(n)})p_2^{(n)}(x_2^{(n)})\Tr\left[\left(W^{(n)}_{x_1^{(n)},x_2^{(n)}} -  A^{(n)}_{x_1^{(n)},x_2^{(n)}}\right)_+\right]\notag\\
								&\hspace{-1em}\leq \sum_{x_1^{(n)}, x_2^{(n)}}p_1^{(n)}(x_1^{(n)})p_2^{(n)}(x_2^{(n)})\Tr\left[\left(W^{(n)}_{x_1^{(n)},x_2^{(n)}} -  B^{(n)}_{x_1^{(n)},x_2^{(n)}}\right)_+\right]\notag\\
								&\hspace{-1em}\leq \sum_{x_1^{(n)}, x_2^{(n)}}p_1^{(n)}(x_1^{(n)})p_2^{(n)}(x_2^{(n)})\Tr\left[W^{(n)}_{x_1^{(n)},x_2^{(n)}}\left\{W^{(n)}_{x_1^{(n)},x_2^{(n)}} >  B^{(n)}_{x_1^{(n)},x_2^{(n)}}\right\}\right],\label{eq:71}
				\end{align}
				where
				\begin{align}
								\hspace{-1em}A^{(n)}_{x_1^{(n)},x_2^{(n)}} = 	e^{n(R_1 - 2 \gamma)} \sigma^{(n)}_{x_2^{(n)}} +e^{n(R_2 - 2 \gamma)} \sigma^{(n)}_{x_1^{(n)}} + e^{n(R_1 + R_2 - 3 \gamma)} \sigma^{(n)},\label{eq:72}\\
								\hspace{-1em}B^{(n)}_{x_1^{(n)},x_2^{(n)}} = 	e^{n(R_1 - 2 \gamma)} \sigma^{(n)}_{x_2^{(n)}} +e^{n(R_2 - 2 \gamma)} \sigma^{(n)}_{x_1^{(n)}} + e^{n(R_1 + R_2 - 4 \gamma)} \sigma^{(n)}. \label{eq:73}
				\end{align}
				Therefore, it follows that
				\begin{align}
								&\hspace{-2em}\varepsilon_n \geq \sum_{x_1^{(n)}, x_2^{(n)}}p_1^{(n)}(x_1^{(n)})p_2^{(n)}(x_2^{(n)}) \Tr\left[W^{(n)}_{x_1^{(n)},x_2^{(n)}}\left\{W^{(n)}_{x_1^{(n)},x_2^{(n)}} \leq B^{(n)}_{x_1^{(n)},x_2^{(n)}} \right\}\right]\notag\\
								&\hspace{14em}-3 e^{-n\gamma}. \label{eq:74}
				\end{align}
Hence, from \eqref{eq:67} and \eqref{eq:74} we have
\begin{align}
				K(R_1 - 2\gamma,R_2-2 \gamma|\vec{p_1},\vec{p_2},\vec{\sigma}) \leq \limsup_{n \to \infty} \varepsilon_n \leq \varepsilon. \label{eq:75}
\end{align}
Since $\gamma$ is arbitrary, \eqref{eq:75} implies that
\begin{align}
				(R_1,R_2)\in\mathrm{Cl}( \{(R_1,R_2)| K(R_1,R_2|\vec{p_1},\vec{p_2},\vec{\sigma}) \leq \varepsilon\}). \label{eq:76}
\end{align}
\end{proof}
We also have
\begin{theo} \label{theo:4}
				\begin{align}
								\hspace{-2em}C^*(\vec{W}) \subset \bigcup_{\vec{p_1},\vec{p_2}} \bigcap_{\vec{\sigma}}\mathrm{Cl}( \{(R_1,R_2)| K^*(R_1,R_2|\vec{p_1},\vec{p_2},\vec{\sigma}) < 1\}). \label{eq:77}
				\end{align}
\end{theo}
\begin{proof}
				 Let $(\cdot)^c$ denote the complement and $\mathcal{R}(\vec{p_1},\vec{p_2},\vec{\sigma}):= \mathrm{Cl}( \{(R_1,R_2)| K^*(R_1,R_2|\vec{p_1},\vec{p_2},\vec{\sigma}) < 1\})$.
				 Suppose that for arbitrary $\vec{p_1}, \vec{p_2}$, there exists $\vec{\sigma}$ satisfying that $(R_1,R_2) \in \mathcal{R}(\vec{p_1}, \vec{p_2},\vec{\sigma})^c$, which means that $(R_1,R_2)$ belongs to the right hand side of \eqref{eq:77}. 
				Then $(R_1 - 2\gamma, R_2- 2\gamma)$ is also in $\mathcal{R}(\vec{p_1},\vec{p_2},\vec{\sigma})^c$ for sufficiently small positive number $\gamma$ since $\mathcal{R}(\vec{p_1},\vec{p_2},\vec{\sigma})^c$ is open.
				This implies that
				\begin{align}
								K^*(R_1-2\gamma,R_2-2\gamma|\vec{p_1},\vec{p_2},\vec{\sigma}) = 1. \label{eq:78}
				\end{align}
				On the other hand, for $\forall (n, M_1^{(n)}, M_2^{(n)}, \varepsilon_n)$-codes which satisfies \eqref{eq:65} and \eqref{eq:66} for $\gamma$ which is used in \eqref{eq:78} and for all sufficiently large $n$, the sequences of the input distributions which are set as \eqref{eq:68} and \eqref{eq:69} are clearly independent.
				Furthermore, for such codes we have \eqref{eq:74}. 

				Hence, from \eqref{eq:74} and \eqref{eq:78} we have
				\begin{align}
								\lim_{n \to \infty} \varepsilon_n = 1. \label{eq:79}
				\end{align}
				This means that $(R_1-\gamma, R_2-\gamma) \in C^*(\vec{W})^c$.
				For arbitrary positive numbers $S_1 < S'_1$ and $S_2 < S'_2$, if $(S_1, S_2) \in C^*(\vec{W})^c$, then clearly $(S'_1, S'_2) \in C^*(\vec{W})^c$ from the definition of $C^*(\vec{W})$.
				Therefore, $(R_1, R_2) \in C^*(\vec{W})^c$.
\end{proof}
The rest of this section is devoted to show how our results lead to the converse parts of classical capacity theorems obtained by Han \cite{Han} \cite{Hanbook}.
First, let $\vec{W_p} = \{(W^{(n)}_{p_1^{(n)}p_2^{(n)}},W^{(n)}_{x^{(n)}_1,p_2^{(n)}},W^{(n)}_{p_1^{(n)},x^{(n)}_2})\}$ be defined as
				\begin{align}
								W^{(n)}_{p_1^{(n)},p_2^{(n)}} &:= \sum_{x^{(n)}_1,x^{(n)}_2}p_1^{(n)}(x_1^{(n)})p_2^{(n)}(x_2^{(n)})W^{(n)}_{x^{(n)}_1,x^{(n)}_2},\label{eq:80}\\
								 W^{(n)}_{x^{(n)}_1,p_2^{(n)}} &:= \sum_{x^{(n)}_2}p_2^{(n)}(x_2^{(n)})W_{x^{(n)}_1,x^{(n)}_2},  \label{eq:81}\\
								 W^{(n)}_{p_1^{(n)},x^{(n)}_2} &:= \sum_{x^{(n)}_1}p_1^{(n)}(x_1^{(n)})W_{x^{(n)}_1,x^{(n)}_2}.\label{eq:82}
				\end{align}
				Then from Theorems \ref{theo:3} and \ref{theo:4}
				we have
				\begin{align}
								\hspace{-2em}C(\varepsilon|\vec{W}) \subset \bigcup_{\vec{p_1},\vec{p_2}} \mathrm{Cl}( \{(R_1,R_2)| K(R_1,R_2|\vec{p_1},\vec{p_2},\vec{W_p}) \leq \varepsilon\})  \label{eq:83}
				\end{align}
				and
				\begin{align}
								\hspace{-2em}C^*(\vec{W}) \subset \bigcup_{\vec{p_1},\vec{p_2}} \mathrm{Cl}( \{(R_1,R_2)| K^*(R_1,R_2|\vec{p_1},\vec{p_2},\vec{W_p}) < 1\}).  \label{eq:84}
				\end{align}
				In the classical case, recalling that the Yagi-Oohama bound implies the Han bound,
				we can easily show that
				\[ K(R_1,R_2|\vec{p_1}, \vec{p_2}, \vec{W_p}) \geq J(R_1,R_2| \mathbf{X_1}, \mathbf{X_2})\] and \[K^*(R_1,R_2|\vec{p_1}, \vec{p_2}, \vec{W_p}) \geq J^*(R_1,R_2| \mathbf{X_1}, \mathbf{X_2}),\]
				where $J$ and $J^*$ are difined in \cite{Han} \cite{Hanbook}. 
				Therefore we have
				\begin{align}
								\{(R_1,R_2)|K \leq \varepsilon\} &\subset \{(R_1,R_2)|J \leq \varepsilon\},\label{eq:85}\\
								\{(R_1,R_2)|K^* < 1\} &\subset \{(R_1,R_2)|J^* < 1 \}. \label{eq:86}
				\end{align}
				Note that Han also proved their direct parts in \cite{Han} \cite{Hanbook}, which establish capacity formulas:
				\begin{align}
								C(\varepsilon|\vec{W}) &= \bigcup_{\vec{p_2},\vec{p_2}}\mathrm{Cl}(\{(R_1,R_2)|J \leq \varepsilon\}),\label{eq:87}\\
								C^*(\vec{W})&= \bigcup_{\vec{p_1},\vec{p_2}}\mathrm{Cl}(\{(R_1,R_2)|J^* < 1 \}), \label{eq:88}
				\end{align}
				although \eqref{eq:88} is not explicitly presented in \cite{Han} \cite{Hanbook}.
				As a consequence, we have
				\begin{align}
								\bigcup_{\vec{p_1},\vec{p_2}}\mathrm{Cl(}\{(R_1,R_2)|K \leq \varepsilon\}) &= \bigcup_{\vec{p_1},\vec{p_2}}\mathrm{Cl}(\{(R_1,R_2)|J \leq \varepsilon\})\label{eq:89}\\
								\bigcup_{\vec{p_1},\vec{p_2}}\mathrm{Cl}(\{(R_1,R_2)|K^* < 1\}) &= \bigcup_{\vec{p_1},\vec{p_2}}\mathrm{Cl}(\{(R_1,R_2)|J^* < 1 \}) \label{eq:90}
				\end{align}
				in the classical case.

				In the quantum case, on the other hand, since we have not proven their direct parts, it is not clear whether Theorem \ref{theo:3} and \ref{theo:4} are tight.
\section{Concluding Remarks} \label{sec:9}
We have discussed lower bounds on the error probability for MACs in several settings.
We have obtained a fundamental inequality in the classical case (Theorem \ref{theo:1}) and in the quantum case (Theorem \ref{theo:2}).
Using the inequality the Yagi-Oohama bound has been generalized and strengthened in several directions and extended to the quantum case.

We have also shown converse results on the $\varepsilon$-capacity region problem and the strong converse region problem for general quantum MACs as applications of the fundamental inequality. 
It however remains to obtain a good upper bound on the error probability in order to determine these regions.

\end{document}